\documentclass{article}
\usepackage{latexsym}
\usepackage{amsmath}
\usepackage{amssymb}
\usepackage{amsfonts}
\usepackage{graphicx}

\newtheorem{remark}{Remark}
\newtheorem{theorem}{Theorem}
\newtheorem{definition}{Definition}
\newtheorem{proposition}{Proposition}
\newenvironment{proof}[1][Proof]{\noindent\textbf{#1.} }{\ \rule{0.5em}{0.5em}}
\usepackage{amssymb}
\usepackage{eurosym}
\usepackage{amsfonts}
\usepackage{amsmath}
\usepackage{graphicx}

\newtheorem{notation}[theorem]{Notation}

\usepackage{latexsym}

\begin{document}

\title{Synthesizing attractors of Hindmarsh-Rose neuronal systems}
\author{Marius-F. Danca\\
\small Dept. of Mathematics and Computer Science\\
\small               Avram Iancu University\\
\small               400380, Cluj-Napoca, Romania\\
\small and \\
\small               Romanian Institute of Science and Technology,\\
\small               400487 Cluj-Napoca, Romania\\
\newline
Qingyun Wang\\
\small Dept. of Dynamics and Control\\
\small               Beihang University, Beijing 100191, China\\
\small and \\
\small               School of Statistics and Mathematics\\
\small               Inner Mongolia Finance and Economics College,
\small Huhhot 010070, \small               China}

\date{}
\maketitle

\begin{abstract}In this paper a periodic parameter switching
scheme is applied to the Hindmarsh-Rose neuronal system to
synthesize certain attractors. Results show numerically, via
computer graphic simulations, that the obtained synthesized
attractor belongs to the class of all admissible attractors for the
Hindmarsh-Rose neuronal system and matches the averaged attractor
obtained with the control parameter replaced with the averaged
switched parameter values. This feature allows us to imagine that
living beings are able to maintain vital behavior while the control
parameter switches so that their dynamical behavior is suitable for
the given environment.
\end{abstract}

\textbf{keywords}: Hindmarsh-Rose model; chaotic attractors; local
attractors; global attractors

\begin{figure*}[ht]
 \begin{center}
\includegraphics[clip,width=1\textwidth]{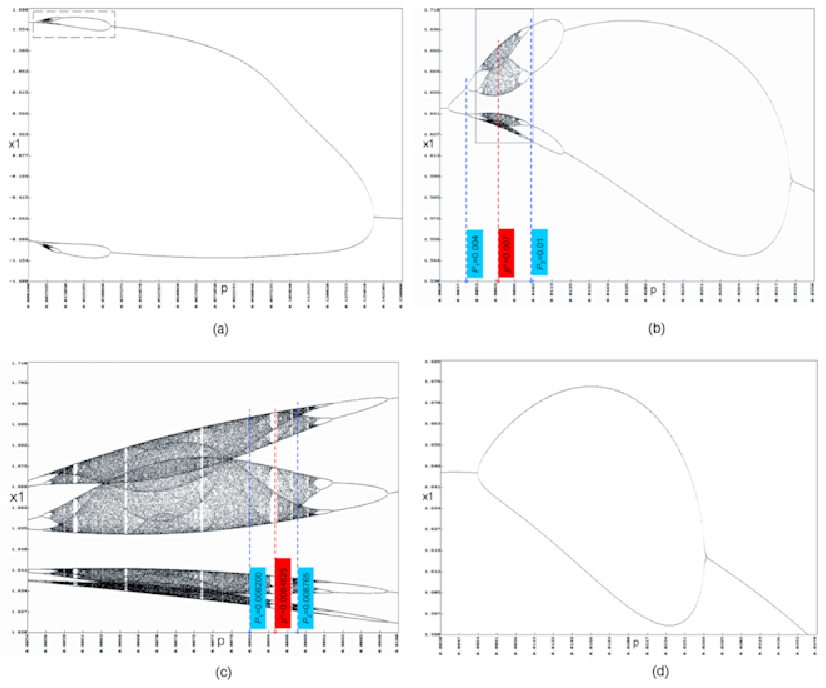} \caption{(a) Bifurcation diagram:
$x_{1}$ vs $p$ with $I=3.4$. (b) Detail of Fig. 1a. (c) Consecutive
detail of Fig. 1b. (d) Bifurcation diagram: $x_{1}$ vs $p$ with
$I=3.5.$}
\label{fig:1}       % Give a unique label
\end{center}
\end{figure*}

\section{Introduction}

Based on the theory of dynamical systems, Hindmarsh and Rose
proposed the phenomenological neuron model, which may be seen either
as a generalization of the Fitzhugh equations or as a simplification
of the physiologically realistic model proposed by Hodgkin and
Huxley \cite{HR,HR2}. The Hindmarsh-Rose (HR) model of neuronal
activity is aimed to study the spiking-bursting behavior of the
membrane potential observed in experiments of a cell in the brain of
the pond snail \cite{HR}. The dynamics of a single HR neuron has
been studied well, and it is illustrated that it can exhibit complex
dynamical behavior including periodic and chaotic spiking (bursting)
motion when certain control parameters of nervous cell models are
varied \cite{HR,HR2,holder1}.

It is well accepted that the HR neuron model is an alternative
candidate for studying the dynamics of neuronal systems since it has
simpler mathematical forms than Chay's and Hodgkin and Huxley's
neuron models \cite{Chay,HH}. Therefore, this model has been used to
study different aspects of neuronal dynamics such as transitions
between different firing regimes \cite{HRf1}, relations between
block structured dynamics and neuronal coding \cite{HRf2}, the
effect of noise on neuronal signal transduction \cite{HRf3} and on
synchronization \cite{HRf4}, the collective dynamics of neuronal
networks \cite{WQ} and the evolution of spiral waves in coupled
Hindmarsh-Rose neurons \cite{MJ}.

However, an important phenomenon in neuron activity is the
transition between different firing patterns. From the viewpoint of
dynamical systems, it is crucial to investigate the transition
mechanism of firing patterns for understanding realistic neuronal
activities. The local stability and the numerical asymptotic
analysis of the Hindmarsh-Rose model are used to investigate the
firing evolution of a single Hindmarsh-Rose neuron \cite{Aziz}. A
simple one-dimensional map method has been applied to the HR neurons
to convert irregular spiking and chaotic bursting to regular beating
and periodic bursting \cite{Chaosc}.

Recently, we developed a parameter switching method to synthesize
attractors of a class of dynamical systems, called hereafter
attractors synthesis (AS) algorithm. This method in fact starts from
a set of given parameter values, and allows us, via periodic or
random parameter-switching, to generate any of the set of all
possible attractors of a class of continuous time dynamical systems
of integer order\cite{Danca},\cite{Danca et al}, or of fractional
order \cite{Danca Kai}, depending linearly on the control parameter.
It has been applied successfully to several dynamical systems such
as Lorenz, Chen, R\"{o}ssler, Lotka-Volterra, minimal networks,
fractional L\"{u} systems and so on. Extending this subject, we will
synthesize attractors of the HR neuronal model, which can exhibit
burst dynamics, via the AS method contributing to a better
understanding of neuronal firing transition.

The paper is organized as follows: the HR model will be described in
Sec. 2; the AS method will be introduced in Sec. 3, and attractors
will be synthesized in Sec. 4. Finally, the paper will end in Sec. 5
with some comments and conclusions.

\section{Description of Hindmarsh-Rose model}

\indent The dynamics of an isolated HR neuron is governed by the
following set of differential equations \cite{HR}:
\begin{equation}%
\begin{array}
[c]{cl}%
\overset{\cdot}{x}_{1}= & bx_{1}^{2}-ax_{1}^{3}+x_{2}-x_{3}+I,\\
\overset{\cdot}{x}_{2}= & c-dx_{1}^{2}-x_{2},\\
\overset{\cdot}{x}_{3}= & p[s(x_{1}-\overline{x}_{1})-x_{3}],
\end{array}
\label{HR}%
\end{equation}

\noindent where $x_{1}$ is the membrane potential, $x_{2}$ is
associated with the fast current, Na$^{+}$, or K$^{+}$ and $x_{3}$
with the slow current, for example, Ca$^{2+}$. The parameters are
\cite{HR}: $a=1,~b=3,~c=1,~d=5,$
$s=4$,~~$\overline{x}_{1}=-1/2(1+\sqrt{5)}\simeq-1.6$. $I$ and $p$
are the control parameters, and $I~$is a slow parameter while $p~$is
a fast parameter. $I$ mimics the membrane input current for
biological neurons; $p$ controls the speed of variation of the slow
variable.

To find the fixed points we have to solve numerically the following equations:
$x_{2}=-5x_{1}^{2}+1,$ \ $x_{3}=4x_{1}+6.4$, and $x_{1}^{3}+2x_{1}^{2}%
+4x_{1}+2=0$ which gives the single real solution $x_{1}=-0.639,~x_{2}%
=-1.041,~x_{3}=3.8442.$ Therefore the system (\ref{HR}) has a single
equilibrium point: $X^{\ast}=\left(  -0.639,-1.041,3.844\right)
^{T}$

Parameter $p~$ is the ratio of time scales between fast dynamics and
slow dynamics. Therefore, it controls the difference between the
slow and the fast dynamics of HR neuron model corresponding to the
difference between fast fluxes of ions across the membrane and slow
ones. Therefore, it is really interesting to investigate the
dynamics of the HR neuron as the parameter $p~$is changed. However,
with $I$ as control parameter, the underlying dynamical system does
not belong to the class of systems where AS can be applied [see
Section 3].

The bifurcation diagram of the first component $x_{1}$ versus $p~$
is shown in Fig. 1a,b,c for $I=3.4$, while that for $I=3.5$ in Fig.
1d. We are interested in the case $I=3.4$ because this case
represents interesting dynamics, where chaos appears in a narrow
range of $p$ with a width of about $3.5\times10^{-3}$. Fig. 1a gives
the bifurcation diagram with respect to the control parameter $p~$
in the range [0.0001, 0.15]. It is shown that as the parameter $p~$
is changed, the HR neuron firstly bifurcates from period-doubling to
chaos, and then, it is stopped via the inverse period-doubling. For
a clear vision, the enlargement of Fig. 1a is shown in Fig. 1b,
which further confirms the above observation. Further enlargement
can clearly guide our assessment as illustrated in Fig. 1c.

More details on the dynamics of the HR model can be found in
\cite{Aziz}

\section{Attractors synthesis algorithm}

\noindent The AS algorithm can be applied to the following general
class of continuous-time autonomous and dissipative dynamical
systems, modeled by the
Initial Value Problem (IVP) \cite{Danca et al}%
\begin{equation}
~S:~\dot{x}=f_{p}(x),\quad x(0)=x_{0}, \label{IVPa}%
\end{equation}
\noindent where $f_{p}\,$ is an $\mathbb{R}^{n}$-valued function with a
bifurcation parameter $p\in\mathbb{R}$, $n\geq3$, and has the expression%
\begin{equation}
f_{p}(x\mathbf{)=}g(x\mathbf{)}+pAx, \label{IVPb}%
\end{equation}
\noindent $g:\mathbb{R}^{n}\longrightarrow\mathbb{R}^{n}~\ $is a
continuous-time nonlinear function, $A~$is a real constant $n\times
n$ matrix,$~x_{0}\in\mathbb{R}^{n}$, and $t\in\lbrack0,T).$

This class of dynamical systems contains the best-known systems such
as Lorenz, Chen, R\"{o}ssler, Lotka-Volterra, minimal networks,
fractional L\"{u} systems and so on (see \cite{Danca}\cite{Danca et
al}\cite{Danca Kai}).

In \cite{Danca et al} it has been shown, via numerical analysis and
computer graphic simulations, that the AS algorithm allows the
synthesis of any attractor of $S$ by parameter switching following
some designed rule.

The AS algorithm can explain what happens with a system when,
intentionally or not, the parameter value switches quickly through a
set of values. Thus, when $p$ is switched following some designed
deterministic \cite{Danca et al} or random \cite{Danca} rule, while
the system evolves in time, an attractor belonging to the set of
attractors is generated (synthesized).

The AS algorithm can be useful in the cases where some desired
control parameter value cannot be directly accessed and we want to
obtain the corresponding attractor.

\noindent Let us next suppose the following assumptions hold.

\noindent\textbf{Assumption} \textbf{A.1. }Throughout, the existence
and uniqueness of solutions of the IVP (\ref{IVPa})-(\ref{IVPb}) are
assumed.

As known, the computer graphic simulations of the numerical integration
results of (\ref{IVPa})-(\ref{IVPb}) can give excellent approximations to the
orbits within the invariant sets \cite{Stuart}. Thus, the orbits which start
near a hyperbolic attractor will stay near and they will be shadowed by orbits
within the attractor. This happens because attractors arise as the limiting
behavior of orbits. Therefore, the shadowing property of hyperbolic sets
\cite{Coombes} enables us to recover long time approximation properties of
numerical orbits such as HR's case.

The AS algorithm consists in using a time varying, or more precisely, a
periodically switching parameter, according to some design rule. It will be
shown, empirically by various experiments, that a desired targeted attractor
can be duly obtained by the proposed switching scheme.

\begin{remark}
The algorithm is robust to some extent: the switching timing and
switching parameter values both allow flexibility. Therefore, they
do not need to be very precisely determined.
\end{remark}
Hereafter the following notations will be used

\begin{notation}
Let us denote by $\mathcal{A}$ the set of all global attractors of
$S$, including attractive stable fixed points, limit cycles and
chaotic (possibly strange) attractors; $\mathcal{P\subset
}\mathbb{R}$ the ordered set of the corresponding admissible values
of $p$~and $\mathcal{P}_{N}=\{p_{1},p_{2},\ldots,p_{N}\}\subset$
$\mathcal{P}\,$ a finite ordered subset of $\mathcal{P}$ containing
$N$ different values $p$, which
determines the set of attractors $\mathcal{A}_{N}=\{A_{p_{1}},A_{p_{2}}%
,\ldots,A_{p_{N}}\}\subset\,\mathcal{A}$.
\end{notation}

\noindent\textbf{Assumptions }

\noindent\textbf{A.2. }$\mathcal{P}$ is considered, as in most
cases, to be composed of
a single real interval and all the values of $\mathcal{P}_{N}=\{p_{1}%
,p_{2},\ldots,p_{N}\}~$for which the system behaves stable and/or
chaotic is assumed to be accessible.

\noindent\textbf{A.3. }$S$ is dissipative i.e.
$\bigtriangledown\cdot f<0,$ where $\bigtriangledown\cdot
f\equiv\sum\limits_{i=1}^{n}\partial f\left(
x_{1},x_{2},\ldots,x_{n}\right)  /\partial x_{i}$ (see e.g.
\cite{Ott}).

Due to the assumed dissipativity, $\mathcal{A}$ is
non-empty\footnote{Attractor sets can exist only for dissipative
systems because shrinking of the volume in phase space for
conservative systems is ruled out by Liouville theorem} and it
follows naturally that for the considered class of systems, a
bijection may be defined between the sets $\mathcal{P}$ and
$\mathcal{A}$. Thus, giving any $p\in\mathcal{P}$, a unique global
attractor is specified, and vice versa.

\begin{remark}
Because in this paper computer simulations are used as the major
analytical tool, the $\omega$-\emph{limit set} (actually, its
approximation \cite{Foias}) is considered after neglecting a
sufficiently long period of transients. Therefore, by
\emph{attractors} (background on the notion of attractor can be
found in \cite{Milnor}) it is appropriate to understand in this
paper the $\omega $-\emph{limit set} obtained by a numerical method
for ODEs with fixed step size $h$ after the transients were
neglected.
\end{remark}

Let $\mathcal{P}_{N}=\{p_{1},p_{2},\ldots,p_{N}\}.$ The AS algorithm relies on
the following deterministic time rule applied repeatedly on $I$%

\begin{equation}
\lbrack(m_{1}h)p_{1},~(m_{2}h)p_{2},\ldots,(m_{N}h)~p_{N}], \label{scheme}%
\end{equation}

\noindent where the weights $m_{i}~$are some positive integers. The
algorithm acts as follow: in the first time subinterval of length
$m_{1}h$, $p$ will have the value $\ p_{1}~$(i.e. the
IVP(\ref{IVPa})-(\ref{IVPb}) will be integrated $m_{1}$ steps for
$p=p_{1}$), for the next $m_{2}^{{}}$ integration steps, $p=p_{2}$
and so on until the $N$-th time subinterval of length $m_{N}^{{}}h$
where $p=p_{N}$ and then the algorithm repeats. In order to simplify
the notation in (\ref{scheme}), for a fixed step size $h$, the
scheme
(\ref{scheme}) will be denoted next%

\begin{equation}
\lbrack m_{1}p_{1},~m_{2}p_{2},\ldots,m_{N}~p_{N}]. \label{schema simpla}%
\end{equation}

For example, the scheme $\left[  p_{3},3p_{1},2p_{2}\right]  $ represents the
infinite sequence of $p:$ $p_{3},3p_{1},2p_{2},p_{3},3p_{1},2p_{2}%
,\ldots~\ \ $which means that while the considered numerical method integrates
(\ref{IVPa})-(\ref{IVPb}), $~p$ switches in each $m_{i}h$ time subinterval. In
other words, the numerical method will integrate (\ref{IVPa})-(\ref{IVPb}) one
step with $p=p_{3},$ then three times with $p=p_{1},$ then two steps with
$p=p_{2}$ and so on.

The AS algorithm for the scheme (\ref{schema simpla}) is presented in Fig. 2%.

\begin{figure}
\[%
\begin{array}
[c]{l}%
t=0\\
repeat\\
~\ \ \ \ for\quad k=1~to~N~do\\
~\ \ \ \ \ \ \ \ \ \ p=p_{k}\\
~~~~~~~~~~\ for~i=1~to~m_{k}~do\\
~~~\ ~~~~~~~~~~~~~integrate~(\ref{IVPa})-(\ref{IVPb})\\
~~~~~~~~~~~~~~~~~t=t+h\\
~~~~~~~~~~~end\\
~~~~~end\\
until~t\geq T
\end{array}
\]
 \begin{center}
\caption{Pseudocode of AS algorithm.}
\label{fig:2}       % Give a unique label
\end{center}
\end{figure}

\begin{notation}
Let us denote the \emph{synthesized} attractor obtained with AS algorithm, via
(\ref{schema simpla}), by $A^{\ast}~$and by $A_{p^{\ast}}$ the \emph{averaged
}attractor corresponding to
\end{notation}

\begin{equation}
p^{\ast}=\frac{\sum\limits_{k=1}^{N}p_{k}m_{k}}{\sum\limits_{k=1}^{N}m_{k}}.
\label{p*}%
\end{equation}

Next, the following notion is introduced

\begin{definition}
Two attractors are considered to be \emph{identical} if their
underlying $\omega$-\emph{limit sets }coincide in the phase space,
the identity being considered from a geometric point of view in the
phase space, aided by computer graphic analysis.
\end{definition}

This identity in the case of chaotic attractors will be considered
only asymptotically since they are fully depicted only after an
infinite time.

\begin{proposition}
The synthesized attractor $A^{\ast}$ and the averaged attractor $A_{p^{\ast}}$
are identical
\end{proposition}

\begin{proof}
Let us consider some subset $\mathcal{P}_{N}$. If we denote $\alpha_{k}%
=m_{k}/\sum\limits_{i=1}^{N}m_{i},$ it is easy to see that
$p^{\ast}$ is a convex combination
$p^{\ast}=\sum\limits_{k=1}^{N}\alpha_{k}p_{k},$ since $\
\sum\limits_{k=1}^{N}\alpha_{k}=1.$ Therefore $p^{\ast}$ belongs
within the interval $\left(  p_{1},\ldots,p_{N}\right)  $, whatever
the values $p_{i}$ are chosen. \noindent Also, taking into account
the bijection between $\mathcal{P}$ and $\mathcal{A}$, we are
entitled to consider that the same convex structure is preserved
into $\mathcal{A}.$ Therefore, for whatever switched values of $p$
in some subset $\mathcal{P}_{N}$, $A^{\ast}$ will belong within the
set $\mathcal{A}_{N}$ $~$considered to be ordered by the mentioned
bijection i.e. $A_{p^{\ast}}\in\left(  A_{p_{1}},A_{p_{N}}\right) $.
\ Next, aided by numerical analysis and via computer graphics, it
can be showed that $A^{\ast}$ is identical (in the sense defined
above) to $A_{p^{\ast}}.$Therefore $A^{\ast}\in\left(
A_{p_{1}},A_{p_{N}}\right)  .$
\end{proof}

Next, we can formulate the main property of the AS algorithm

\begin{proposition}
For whatever considered set $\mathcal{P}_{N}$, the AS algorithm generates an
attractor $A^{\ast}~$which belongs to $\left(  A_{p_{1}},A_{p_{N}}\right)  .$
\end{proposition}

\begin{remark}
The AS algorithm is useful especially when we want to obtain some
attractor $\overline{A}_{\overline{p}}~$even the underlying value
$\overline{p}~$cannot be accessible.

In this case $\overline{A}$ can be synthesized by choosing a
corresponding set $\left( p_{1},\ldots,p_{N}\right)
\ni\overline{p}$ but
$\overline{p}\notin\mathcal{P}_{N}=\{p_{1},\ldots,p_{N}\},$ and a
corresponding scheme (\ref{schema simpla}).
\end{remark}

The initial conditions play an important role since for a specific value
$p\in\mathcal{P}$ there is a single global attractor but which could be
composed by several local attractors (see e.g. \cite{Milnor},\cite{Hirsch1}%
,\cite{Hirsch2},\cite{Kapitan}). For example, for the Lorenz system
for the control parameter $r=2.5$~there are three local attractors:
the origin and two symmetrical fixed points
$X_{1},_{2}(\pm2,~\mp2,~1.5$) while for $r=28$ there is a single
local attractor which is global too, the known Lorenz strange
attractor. To avoid these possible difficulties, all the computer
simulations for $A^{\ast}$ and $A_{p^{\ast}}$ for a particular case,
start from the same initial conditions.

\begin{figure*}[!ht]

\includegraphics[clip,width=1\textwidth]{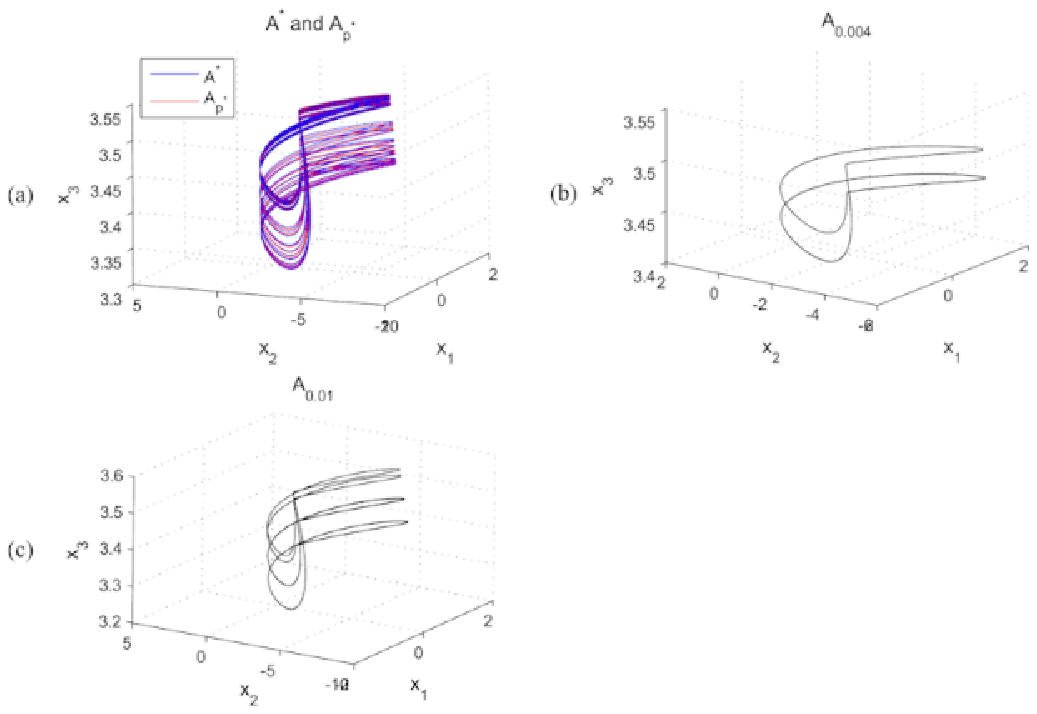}
\caption{AS algorithm with scheme $[1p_{1},1p_{2}]$ with
$p_{1}=0.004$ and $p_{2}=0.01;$ a) The synthesized and averaged
attractors $A^{\ast}$ and $A_{p^{\ast}}~$respectively, superimposed
for $p^{\ast}=0.007~$; b) Attractor $A_{0.004\text{ }};$ c)
Attractor $A_{0.01}.$}
\label{fig:3}       % Give a unique label

\end{figure*}

In order to see how the AS algorithm works, let us consider that we
want to synthesize the attractor $\overline{A}_{\overline{p}}.$

\noindent Then we must choose $N$ and the set $\mathcal{P}_{N}$ such
that $\overline{p}\in\left( p_{1},\ldots,p_{N}\right)  $ (
$\overline{p}~$can be equal or not to one of the elements $p_{i},$
for $i=2,\ldots N-1).$ Next, choosing empirically the scheme
(\ref{schema simpla}), such that the right-hand side of (\ref{p*})
gives $\overline{p},$ the initial value problem is integrated and
finally $\overline{A}_{\overline{p}}$ is obtained.

\noindent To underline the identity between $A^{\ast}$ and
$A_{p^{\ast}}~$histograms and Poincar\'{e} sections besides the
phase plots were drawn after the transients were neglected.

\begin{figure*}[!ht]

\includegraphics[clip,width=0.53\textwidth]{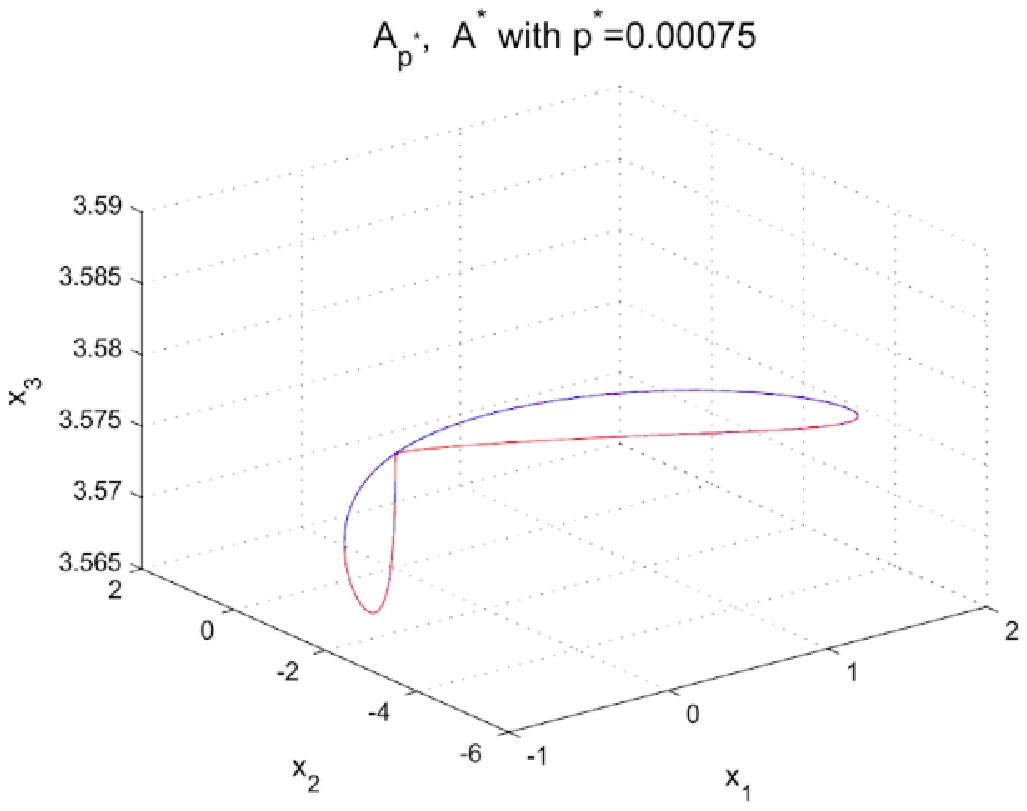}
\caption{AS algorithm for $N=10,$ $p_{i}=0.0002+i\times0.0001~$and
$m_{i}=1,$ for $i=1,2,\ldots,10;$ $A^{\ast}$ and
$A_{p^{\ast}},~$with $p^{\ast}=0.00075,$ are superimposed.}
\label{fig:4}       % Give a unique label

\end{figure*}

\section{Synthesis of HR attractors}

\indent Real neurons often display high nonlinearity, which has been
shown in many experiments and confirmed by numerical simulation of
many neuron models, including the HR and Chay neuron
models\cite{HR,Chay}. Chaos is a universal phenomenon in certain
neurons, such as those in the human brain where, often, the
information is decoded and transduced by means of chaos. It is very
important to study the chaos of neuron from different aspects. Next,
we will focus on the neuronal firing pattern when $p$ is switched.

\noindent Firstly, it is easy to prove the following
\begin{proposition}
The HR system (\ref{HR}) with $p$ considered as control parameter, belongs to
the class of systems modeled by the IVP (\ref{IVPa})-(\ref{IVPb}).

\begin{proof}
Choosing the following substitution\newline%
\[
y_{1}=x_{1}-\overline{x}_{1}%
\]
\newline and replacing finally $y_{1}$ with $x_{1},$ (\ref{HR})
becomes\newline%
\begin{align}
\overset{.}{x}_{1}  &  =a_{1}x_{1}^{3}+b_{1}x_{1}^{2}+c_{1}x_{1}+d_{1}%
x_{2}+e_{1}x_{3}+f_{1}+I,\label{modified IVP}\\
\overset{.}{x}_{2}  &  =a_{2}x_{1}^{2}+b_{2}x_{1}+c_{2}x_{2}+d_{2},\nonumber\\
\overset{.}{x}_{3}  &  =p(sx_{1}-x_{3}),\nonumber
\end{align}
\newline where\newline\newline%
\begin{align*}
a_{1}  &  =-a,\\
b_{1}  &  =b-3a\overline{x}_{1},\\
c_{1}  &  =\overline{x}_{1}\left(  2b-3a\overline{x}_{1}\right)  ,\\
d_{1}  &  =1,\\
e_{1}  &  =-1,\\
f_{1}  &  =-a\overline{x}_{1}^{3}+b\overline{x}_{1}^{2},\\
a_{2}  &  =-d,\\
b_{2}  &  =-2d\overline{x}_{1},\\
c_{2}  &  =-1,\\
d_{2}  &  =d\overline{x}_{1}^{2}.
\end{align*}
\newline Thus, the right hand side of (\ref{modified IVP}) can be written
following (\ref{IVPb}) with%
\[
g(x)=\left(
\begin{array}
[c]{c}%
a_{1}x_{1}^{3}+b_{1}x_{1}^{2}+c_{1}x_{1}+d_{1}x_{2}+e_{1}x_{3}+f_{1}+I\\
a_{2}x_{1}^{2}+b_{2}x_{1}+c_{2}x_{2}+d_{2}\\
0
\end{array}
\right)
\]
\[
\text{ }A=\left(
\begin{array}
[c]{ccc}%
0 & 0 & 0\\
0 & 0 & 0\\
s & 0 & -1
\end{array}
\right) \text{ \rm and } x=\left(
\begin{array}
[c]{c}%
x_{1}\\
x_{2}\\
x_{3}%
\end{array}
\right)  .
\]
\end{proof}
\end{proposition}

\begin{figure*}[!ht]
 \begin{center}
\includegraphics[clip,width=0.91\textwidth]{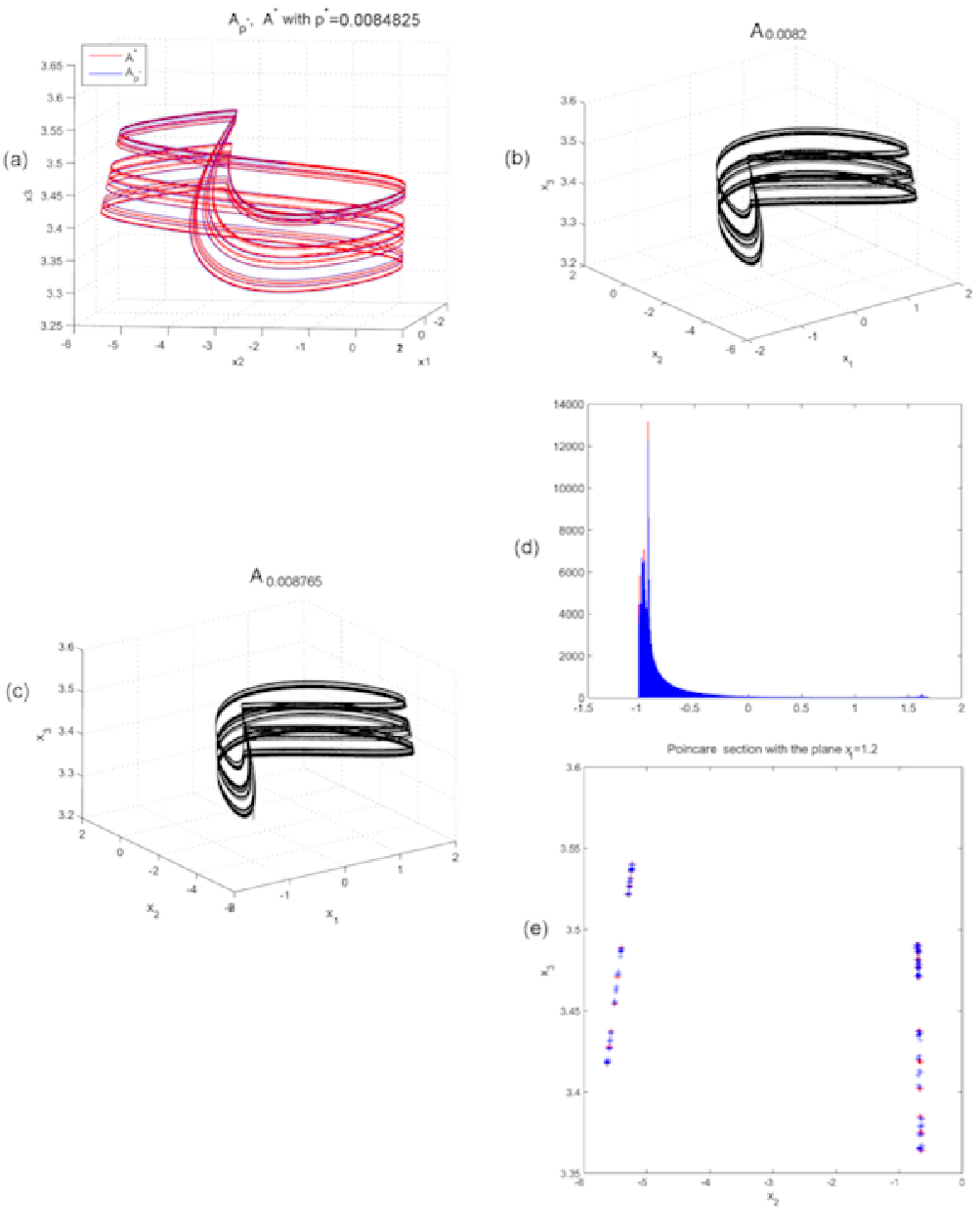}
\caption{AS algorithm for $N=2$ and
$\mathcal{P}_{N}=\{0.0082,0.008765\}$ and $m_{1}=m_{2}=1;$ a)
$A^{\ast}$ and $A_{p^{\ast}},~$with $p^{\ast}=0.0084825~,$
superimposed; b) $A_{0.0082};$ c) $A_{0.008765};$ d) Superimposed
histograms of $A^{\ast}~$and $A_{p^{\ast}};$ e) Superimposed
Poincar\'{e} sections of $A^{\ast}~$and $A_{p^{\ast}}~$for $\
x_{1}=1.2.$}
\label{fig:5}       % Give a unique label
\end{center}
\end{figure*}

Because, from a numerical point of view, the mathematical model
(\ref{HR}) presents a more accessible form than (\ref{modified
IVP}),\ next we will work with (\ref{HR}).

The HR system verifies the assumptions A1, A2 and also A3 for large
parameters range taking into account that $\bigtriangledown\cdot f=-3ax_{1}^{2}%
+2bx_{1}-1-p.$ Therefore, one can apply the AS algorithm to the HR
system (\ref{HR}). For this purpose we considered the Standard
Runge-Kutta numerical method for ODEs with the fixed step size
$h=0.005$ and, as stated above, $I=3.4.$

\begin{figure*}[!t]
\includegraphics[clip,width=0.95\textwidth]{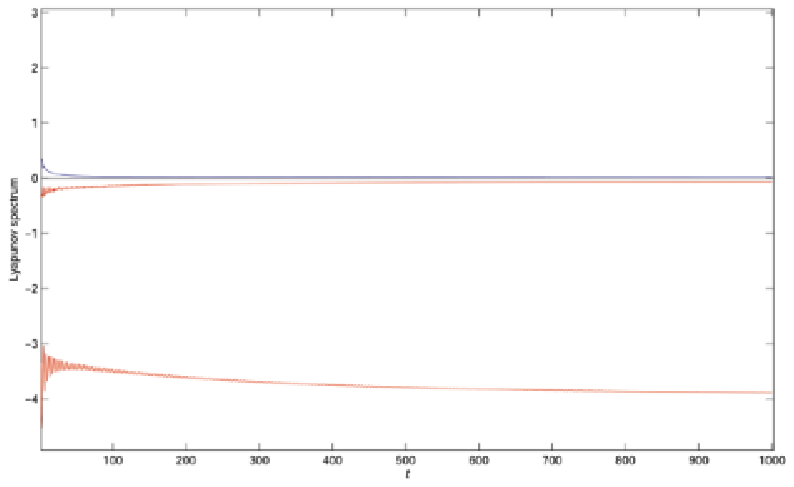}
\caption{Temporal evolution of Lyapunov exponents for
$p=0.0084825.$}
\label{fig:6}       % Give a unique label
\end{figure*}

For the sake of the brevity, in what follows, only the most relevant
cases are considered. Let us firstly consider the case $N=2$ and
$\mathcal{P}_{N}=\{0.004,0.01\}$
(see Fig. 1b) with $m_{1}=m_{2}=1.$ Then, using the scheme $[1p_{1}%
,1p_{2}]$ while integrating the initial value problem (\ref{HR}),
finally, even the attractors $A_{0.004\text{ }}$and $\ A_{0.01}$ are
stable limit cycles, the obtained synthesized attractor $A^{\ast}$
is chaotic (Fig. 3a). As can be seen from Fig. 3a, the synthesized
attractor $A^{\ast}$ (blue) and
the averaged attractor $A_{p^{\ast}}$ (red) for $p^{\ast}%
=(0.004+0.01)/2=0.007~$ given by (\ref{p*}) coincide. It can be seen
that $p^{\ast}\in(p_{1},p_{2})$ (Fig. 1b) and also $A^{\ast}$ is
situated between the corresponding attractors $A_{p_{1}}$ and
$A_{p_{2}}.$

Next, if we use $N=10$ values for $p$ defined as follows: $p_{i}%
=0.0002+i\times0.0001~$and $m_{i}=1,$ for $i=1,2,\ldots,10,$ the
synthesized attractor $A^{\ast}$ is a stable limit cycle which
coincides with the averaged attractor $A_{p^{\ast}}$ with
$p^{\ast}=0.00075~$(Fig. 4).

An interesting case appears for $N=2$ and $\mathcal{P}_{N}%
=\{0.0082,0.008765\}$ and $m_{1}=m_{2}=1,$ where $p_{1},$ and
$p_{2}$ are chosen in a very narrow periodic window (its width is
about $1E-~4~$, see Fig. 1c, Fig. 5b and Fig. 5c). The obtained
attractor, $A^{\ast},$ is identical to $A_{p^{\ast}}$ for
$p^{\ast}=0.0084825~$(Fig. 5 a). We found numerically that the fixed
point $X^{\ast}~$is unstable for $p=p^{\ast}~$since the three
eigenvalues are: $-6.266,0.179$ and $0.02$ and, as showed in Fig. 6,
the Lyapunov spectrum has two negative exponents while the maximum
one is positive (approximately zero). All these lead us to consider
that $A_{p^{\ast}}$ (and consequently $A^{\ast})$ is a stable limit
cycle. The apparent difference between the attractors in this case
is due to the very small difference between $p_{1}$ and $p_{2}.$

\section{\smallskip Conclusions}

\noindent In summary, we have investigated the synthesis of
attractors of the HR neuronal system by means of the AS method. It
is shown that when we choose the slow parameter $p~$ as the control
parameter, the HR neuronal system belongs to a class of systems,
where we can apply the proposed switching method. Consequently, we
concluded that every attractor can be synthesized by the proposed
periodic parameter switching scheme in the HR neuron model. It is
accepted that neurons can code information by means of firing
patterns, which depends on the variation of key parameters of
neuronal systems. Hence, the results presented in this paper may be
instructive in understanding the implications of realistic neuronal
dynamics.

{\bf Acknowledgements} \,\, This work was supported by the National Science Foundation of China (Fund Nos. 10972001, 10702023 and 10832006).

\end{document}